\documentclass[twocolumn,superscriptaddress,
showpacs,preprintnumbers,amsmath,amssymb]{revtex4}

\usepackage{amsthm}
\usepackage{epsfig}
\usepackage{multirow}
\newtheorem{Thm}{Theorem}

\newcommand{\bra}[1]{{\left\langle #1 \right|}}
\newcommand{\ket}[1]{{\left| #1 \right\rangle}}

\newcommand{\Z}{\mbox{$\mathbb Z$}}

\newcommand{\T}{\mbox{$\mathrm{tr}$}}

\begin{document}
\title{Quantum states for perfectly secure secret sharing}

\author{Dong Pyo Chi}
\affiliation{
 Department of Mathematical Sciences,
 Seoul National University, Seoul 151-742, Korea
}
\author{Jeong Woon Choi}
\affiliation{
 Department of Mathematical Sciences,
 Seoul National University, Seoul 151-742, Korea
}
\author{Jeong San Kim}
\affiliation{
 Institute for Quantum Information Science,
 University of Calgary, Alberta T2N 1N4, Canada
}
\author{Taewan Kim}
\affiliation{
 Department of Mathematical Sciences,
 Seoul National University, Seoul 151-742, Korea
}
\author{Soojoon Lee}
\affiliation{
 Department of Mathematics and Research Institute for Basic Sciences,
 Kyung Hee University, Seoul 130-701, Korea
}

\date{\today}

\begin{abstract}
In this work, we investigate what kinds of quantum states are feasible
to perform perfectly secure secret sharing,
and present its necessary and sufficient conditions.
We also show that the states are bipartite distillable for all bipartite splits,
and hence the states could be distillable into
the Greenberger-Horne-Zeilinger state.
We finally exhibit a class of secret-sharing
states, which have an arbitrarily small amount of
bipartite distillable entanglement for a certain split.
\end{abstract}

\maketitle

{\em Introduction.---}
Entanglement has been considered as one of
the most crucial resources for quantum communication, which has
been shown to be perfectly secure against any interior/exterior
eavesdropper. The perfect security seems to be due to the pure
entanglement. However, in the case of the quantum key
distribution, appropriately defining its perfect security as in
Refs.~\cite{HHHO1,HHHO2}, we can see that only pure entanglement
does not guarantee the perfect security.

We here focus on another quantum communication protocol, the
quantum secret sharing of classical information, originally
presented by Hillery {\em et al.}~\cite{HBB}. Our question is what
kinds of quantum states are feasible to perform the perfectly
secure secret sharing (PSSS). In order to answer this question,
first of all, it is required to present the conditions for the
PSSS.

One of the most important problems in secret sharing of classical information is
how to share random bits securely
between one dealer, Alice and other players, Bob and Charlie.
\begin{table}[h]\label{Table00}
\caption{Random bits for secret sharing}
\begin{center}
\begin{tabular}{c|c|c} \hline\hline
Alice & Bob & Charlie \\ \hline
\multirow{2}{*}{0} & 0 & 0 \\
& 1 & 1 \\ \hline
\multirow{2}{*}{1} & 0 & 1 \\
& 1 & 0 \\ \hline\hline
\end{tabular}
\end{center}
\end{table}
If each participant would securely share one of the random bit sequences
as in TABLE~I, 
then Alice could secretly make Bob and Charlie share her secret bit.

Thus, for the PSSS, the two following conditions must be
satisfied: (i)~Probability distributions of all participants'
secret bits should be unbiased and perfectly correlated, that is,
if we let $p_{ijk}$ be the probability that Alice, Bob, and
Charlie get the random bits $i$, $j$, and $k$, respectively, then
$p_{000}=p_{011}=p_{101}=p_{110}=1/4$ and $p_{ijk}=0$ for other
$i$, $j$, $k$. (ii)~Eavesdropper should not be able to obtain any
information about participants' secret bits.

In this work, according to the two above conditions, we show that
$\rho_{ABCA'B'C'}$ is a quantum state for the PSSS if and only if
it is of the form
\begin{equation}
\frac{1}{4}\sum_{ i+j+k \equiv 0 \pmod{2} \atop i'+j'+k' \equiv 0 \pmod{2} }
\ket{ijk}_{ABC}\bra{i'j'k'}\otimes U_{ijk}\rho_{A'B'C'} U_{i'j'k'}^{\dagger},
\label{eq:PQSS_state}
\end{equation}
where $\rho_{A'B'C'}$ is a state on subsystem $A'B'C'$,
and $U_{ijk}$'s are unitary operators.
We call this form of states in~(\ref{eq:PQSS_state}) the {\em secret-sharing states}.

We also show that
the states are bipartite distillable for all bipartite splits.
From the results of D\"{u}r {\em et al.}~\cite{DCT}
we can readily derive the fact that
if any $n$-qubit state has negative partial transposition for all bipartite splits
then it is distillable into the Greenberger-Horne-Zeilinger (GHZ)~\cite{GHZ} state.
Hence, the secret-sharing states could be also distillable into the GHZ state.

Furthermore, we show that our results can be generalized into
multipartite cases, that is, $\rho_n$ is an $n$-qubit state for
the PSSS consisting of one dealer and $n-1$ players if and only if
it is of the form
\begin{equation}
\frac{1}{2^{n-1}}\sum_{I, J \in \Z_2^n \atop \mathrm{even~parity} }
\ket{I}_{A_1A_2\cdots A_n}\bra{J}\otimes U_{I}\rho_{A_1'A_2'\cdots A_n'} U_{J}^{\dagger},
\label{eq:nPQSS_state}
\end{equation}
where $\rho_{A_1'A_2'\cdots A_n'}$ is a state on subsystem $A_1'A_2'\cdots A_n'$,
and $U_{I}$'s are unitary operators,
and that $\rho_n$ is bipartite distillable for its all bipartite splits.
Hence, as in the three-party case,
$\rho_n$ could be distillable into the $n$-qubit GHZ state.


{\em Secret-Sharing States.---} 
We first provide necessary and sufficient conditions for a state to perform
the PSSS. Let $A$, $B$, and $C$ be qubit systems, and $A'$, $B'$,
and $C'$ be of arbitrary dimensions. Here $AA'$, $BB'$, and $CC'$
are Alice's, Bob's, and Charlie's systems, respectively. Then we
obtain the following theorem.

\begin{Thm}\label{Thm:SS_states}
Any state is a quantum state for the 3-party PSSS if and only if
it is a secret-sharing state of the form in (\ref{eq:PQSS_state}).
\end{Thm}
\begin{proof}
We first assume that $\rho_{ABCA'B'C'}$ is a quantum state for the
PSSS, and let $\ket{\Psi}_{ABCA'B'C'E}$ be its purification as
follows:
\begin{equation}
\ket{\Psi}_{ABCA'B'C'E}=
\sum_{i,j,k}\sqrt{p_{ijk}}\ket{ijk}_{ABC} \ket{\Psi_{ijk}}_{A'B'C'E},
\label{eq:purification}
\end{equation}
where $E$ is the system of the eavesdropper, Eve.
Since probability distributions of all participants' secret bits
are unbiased and perfectly correlated,
it is clear that
$p_{000}=p_{011}=p_{101}=p_{110}=1/4$ and $p_{ijk}=0$ for other $i$, $j$, $k$.
Thus, the state $\ket{\Psi}$ becomes
\begin{equation}
\frac{1}{2}\sum_{i+j+k\equiv 0 \pmod{2}}\ket{ijk}_{ABC} \ket{\Psi_{ijk}}_{A'B'C'E}.
\label{eq:Psi2}
\end{equation}
For each of all participants' measurement result $ijk$,
let $\rho_{ijk}^E$ be Eve's state after the measurement.
Then $\rho_{ijk}^E=\T_{A'B'C'}\left(\ket{\Psi_{ijk}}\bra{\Psi_{ijk}}\right)$.
Since Eve cannot obtain any information about participants' secret bit at all,
we have $\rho_{000}^E=\rho_{011}^E=\rho_{101}^E=\rho_{110}^E$.
For each $i$, $j$, $k$,
let $\rho_{ijk}^E=\sum_{l}\lambda_l \ket{\phi_l}\bra{\phi_l}$
be its spectral decomposition.
Then it follows from Gisin-Hughston-Jozsa-Wootters (GHJW) theorem~\cite{GHJW}
that for each $i$, $j$, $k$,
there are unitary operators $U_{ijk}$ on the system $A'B'C'$ such that
\begin{equation}
\ket{\Psi_{ijk}}=\sum_{l}\sqrt{\lambda_l}U_{ijk}\ket{\psi_l}_{A'B'C'}\ket{\phi_l}_E,
\label{eq:Psi_ijk}
\end{equation}
where $\ket{\psi_{l}}$ forms an orthonormal set for the system $A'B'C'$.
Hence, by Eqs. (\ref{eq:Psi2}) and (\ref{eq:Psi_ijk}),
$\rho_{ABCA'B'C'}$ is of the form
\begin{equation}
\frac{1}{4}\sum_{ i+j+k \equiv 0 \pmod{2} \atop i'+j'+k' \equiv 0 \pmod{2} }
\ket{ijk}_{ABC}\bra{i'j'k'}\otimes U_{ijk}\rho_{A'B'C'} U_{i'j'k'}^{\dagger},
\label{eq:PQSS_state2}
\end{equation}
where $\rho_{A'B'C'}=\sum_l \lambda_l \ket{\psi_l}\bra{\psi_l}$.

Conversely, we now assume that a given state $\rho_{ABCA'B'C'}$ is
of the form~(\ref{eq:PQSS_state}). Then since for the
probabilities $p_{ijk}$ that participants get the bits $ijk$
$p_{000}=p_{011}=p_{101}=p_{110}=1/4$ and $p_{ijk}=0$ for other
$i$, $j$, $k$, probability distributions for the secret bits are
clearly unbiased and perfectly correlated. Thus, it suffices to
show that Eve cannot any information about the secret bits, that
is, $\rho_{000}^E=\rho_{011}^E=\rho_{101}^E=\rho_{110}^E$, where
$\rho_{ijk}^E$ is Eve's state when participants' measurement
result is $ijk$. For convenience, we consider the following block
matrix form of $\rho_{ABCA'B'C'}$:
\begin{equation}
\frac{1}{4}
\left(
\begin{array}{cccc|cccc}
  X_{000,000} & 0 & 0 & X_{000,011} & 0 & X_{000,101} & X_{000,110} & 0 \\
  0 & 0 & 0 & 0 & 0 & 0 & 0 & 0 \\
  0 & 0 & 0 & 0 & 0 & 0 & 0 & 0 \\
  X_{011,000} & 0 & 0 & X_{011,011} & 0 & X_{011,101} & X_{011,110} & 0 \\
  &&&&&&& \\ \hline
  &&&&&&& \\
  0 & 0 & 0 & 0 & 0 & 0 & 0 & 0 \\
  X_{101,000} & 0 & 0 & X_{101,011} & 0 & X_{101,101} & X_{101,110} & 0 \\
  X_{110,000} & 0 & 0 & X_{110,011} & 0 & X_{110,101} & X_{110,110} & 0 \\
  0 & 0 & 0 & 0 & 0 & 0 & 0 & 0 \\
\end{array}
\right),
\label{eq:matrix_form}
\end{equation}
where $X_{ijk,i'j'k'}=U_{ijk} \rho_{A'B'C'} U_{i'j'k'}^{\dagger}$.
Then we can readily obtain that
the trace norm of $X_{ijk,i'j'k'}$ is one, that is, $\|X_{ijk,i'j'k'}\|_1=1$.
Let the state in Eq.~(\ref{eq:Psi2}) be the purification of $\rho_{ABCA'B'C'}$.
Then we have
$X_{ijk,i'j'k'}=\T_E\left(\ket{\Psi_{ijk}}\bra{\Psi_{i'j'k'}}\right)$.
It follows from straightforward calculations that
\begin{equation}
\|X_{ijk,i'j'k'}\|_1=\T\left|\sqrt{\rho_{ijk}^E}\sqrt{\rho_{i'j'k'}^E}\right|
=F(\rho_{ijk}^E,\rho_{i'j'k'}^E),
\label{eq:X}
\end{equation}
where $F$ is the fidelity.
Since $F(\rho_{ijk}^E,\rho_{i'j'k'}^E)=1$ for every $i$, $j$, $k$,
the proof is completed.
\end{proof}
We remark that, as seen in the proof of Theorem~\ref{Thm:SS_states},
in order to prove its converse,
it is sufficient to use
that the trace norms of three well-chosen off-diagonal blocks are one,
for example, $\|X_{000,011}\|_1=\|X_{011,101}\|_1=\|X_{101,110}\|_1=1$.
Moreover, any block matrix of the form in~(\ref{eq:matrix_form})
whose three well-chosen off-diagonal blocks have trace norm 1/4
forms a secret state as follows.

\begin{Thm}\label{Thm:block_SS}
$\sigma_{ABCA'B'C'}$ is a state which can be expressed as the following block-matrix form:
\begin{equation}
\sigma_{ABCA'B'C'} =
\left(%
\begin{array}{cccc|cccc}
  \star & 0 & 0 & X & 0 & \star & \star & 0 \\
  0 & 0 & 0 & 0 & 0 & 0 & 0 & 0 \\
  0 & 0 & 0 & 0 & 0 & 0 & 0 & 0 \\
  \star & 0 & 0 & \star & 0 & Y & \star & 0 \\ \hline
  0 & 0 & 0 & 0 & 0 & 0 & 0 & 0 \\
  \star & 0 & 0 & \star & 0 & \star & Z & 0 \\
  \star & 0 & 0 & \star & 0 & \star & \star & 0 \\
   0 & 0 & 0 & 0 & 0 & 0 & 0 & 0 \\
\end{array}
\right),
\label{eq:block_sigma}
\end{equation}
where $\| X \|_{1}=\| Y \|_{1}=\| Z \|_{1}=1/4$
if and only if $\sigma_{ABCA'B'C'}$ is a secret-sharing state.
\end{Thm}
\begin{proof}
Since any secret-sharing state is of the form~(\ref{eq:block_sigma}),
it suffices to show that $\sigma_{ABCA'B'C'}$ in~(\ref{eq:block_sigma})
is a secret-sharing state.
For each $ijk$,
let $p_{ijk}$ be the trace of $(ijk,ijk)$ block entry of $\sigma_{ABCA'B'C'}$.
Then
\begin{eqnarray}
\ket{\Psi}&=&
\sqrt{p_{000}}\ket{000}_{ABC} \ket{\Psi_{000}}_{A'B'C'E}
\nonumber \\
&&+\sqrt{p_{011}}\ket{011}_{ABC} \ket{\Psi_{011}}_{A'B'C'E}
\nonumber \\
&&+\sqrt{p_{101}}\ket{101}_{ABC} \ket{\Psi_{101}}_{A'B'C'E}
\nonumber \\
&&+\sqrt{p_{110}}\ket{110}_{ABC} \ket{\Psi_{110}}_{A'B'C'E}
\label{eq:purification_sigma}
\end{eqnarray}
is its purification,
and hence
we obtain
\begin{eqnarray}
X&=&\sqrt{p_{000} p_{011}}\T_E\left(\ket{\Psi_{000}}\bra{\Psi_{011}}\right),
\nonumber\\
Y&=&\sqrt{p_{011} p_{101}}\T_E\left(\ket{\Psi_{011}}\bra{\Psi_{101}}\right),
\nonumber\\
Z&=&\sqrt{p_{101} p_{110}}\T_E\left(\ket{\Psi_{101}}\bra{\Psi_{110}}\right).
\label{eq:XYZ}
\end{eqnarray}
As in the proof of Theorem~\ref{Thm:SS_states}, we have
\begin{eqnarray}
\|X\|_1&=&\sqrt{p_{000} p_{011}}F\left(\rho_{000}^E,\rho_{011}^E\right),
\nonumber\\
\|Y\|_1&=&\sqrt{p_{011} p_{101}}F\left(\rho_{011}^E,\rho_{101}^E\right),
\nonumber\\
\|Z\|_1&=&\sqrt{p_{101} p_{110}}F\left(\rho_{101}^E,\rho_{110}^E\right),
\label{eq:XYZ}
\end{eqnarray}
where $\rho_{ijk}^E=\T_{A'B'C'}\left(\ket{\Psi_{ijk}}\bra{\Psi_{ijk}}\right)$.
Since $\| X \|_{1}=\| Y \|_{1}=\| Z \|_{1}=1/4$,
we have the following inequalities:
\begin{eqnarray}
\frac{p_{000}+p_{011}}{2}&\ge&\sqrt{p_{000} p_{011}}\ge \frac{1}{4},
\nonumber\\
\frac{p_{011}+p_{101}}{2}&\ge&\sqrt{p_{011} p_{101}}\ge \frac{1}{4},
\nonumber\\
\frac{p_{101}+p_{110}}{2}&\ge&\sqrt{p_{101} p_{110}}\ge \frac{1}{4}.
\label{eq:ineq}
\end{eqnarray}
It follows from the fact $p_{000}+p_{011}+p_{101}+p_{110}=1$ that
$p_{000}=p_{011}=p_{101}=p_{110}=1/4$ and
$F\left(\rho_{000}^E,\rho_{011}^E\right)
=F\left(\rho_{011}^E,\rho_{101}^E\right)
=F\left(\rho_{101}^E,\rho_{110}^E\right)=1$. This implies that
$\sigma_{ABCA'B'C'}$ is a state for the PSSS. Therefore, it is a
secret-sharing state by Theorem~\ref{Thm:SS_states}.
\end{proof}


We note that every private state is distillable~\cite{HA}.
By employing this note, we now show that every secret-sharing state is
bipartite distillable for its all bipartite splits.

\begin{Thm}\label{Thm:distillability}
Let $\rho$ be a secret-sharing state for the 3-party secret sharing.
Then $\rho$ is bipartite distillable for its all bipartite splits
of the 3 parties.
\end{Thm}
\begin{proof} By Theorem~\ref{Thm:SS_states},
$\rho$ can be expressed as the form of (\ref{eq:PQSS_state})
for some state $\rho_{A'B'C'}$ and unitary operators $U_{ijk}$.
Let $\mathrm{CNOT}_{ij}$ be the controlled-NOT operation
such that $i$ and $j$ represent its control system and target system, respectively.
Then applying $\mathrm{CNOT}_{BC}$ to $\rho$ and
performing the projective measurement on system $B$
with respect to the standard basis $\{\ket{0},\ket{1}\}$,
when the measurement result is $r$,
the resulting state becomes a private state,
\begin{equation}
\frac{1}{2}\sum_{i,j=0}^1 \ket{ii}_{AC}\bra{jj}\otimes U_{iri}\rho_{A'B'C'}U_{jrj}^\dagger,
\label{eq:private}
\end{equation}
which is distillable~\cite{HA}.
Thus, a given $\rho$ is bipartite distillable for the split \mbox{$AA'$-$BB'CC'$}.
Similarly, we can show that
$\rho$ is bipartite distillable for the splits \mbox{$BB'$-$CC'AA'$} and \mbox{$CC'$-$AA'BB'$}.
\end{proof}

{\em Generalization into multipartite cases.---} 
We now generalize our results
into multipartite cases.

\begin{Thm}\label{Thm:nSS_states}
$\rho_n$ is a quantum state for the $n$-party PSSS
consisting of one dealer and $n-1$ players
if and only if it is a secret-sharing state of the form in~(\ref{eq:nPQSS_state}).
\end{Thm}
\begin{proof}
Let $\ket{\Psi}$ be a purification of $\rho_n$ as follows:
\begin{equation}
\ket{\Psi}=
\sum_{I\in\Z_2^n}\sqrt{p_{I}}\ket{I}_{A_1A_2\cdots A_n} \ket{\Psi_{I}}_{A_1'A_2'\cdots A_n'E}.
\label{eq:n_purification}
\end{equation}
As in the case of the 3-party case,
it is clear that
$p_I=1/2^{n-1}$ for all $I$ with even parity
and $p_{J}=0$ for other $J$.
For each of all participants' measurement result $I$,
Eve's state after the measurement, $\rho_{I}^E$ becomes
$\rho_{I}^E=\T_{A_1'A_2'\cdots A_n'}\left(\ket{\Psi_{I}}\bra{\Psi_{I}}\right)$.
Since Eve cannot obtain any information about participants' secret bit at all,
all $\rho_{I}^E$'s are the same.
Thus, by GHJW theorem
there are a state $\rho_{A_1'A_2'\cdots A_n'}$ and
unitary operators $U_{I}$ on the system $A_1'A_2'\cdots A_n'$ such that
$\rho_n$ is of the form in~(\ref{eq:nPQSS_state})
and $\rho_{A_1'A_2'\cdots A_n'}$ has the same spectrum as $\rho_{I}^E$.

Conversely, assuming that
a given state $\rho_n$ is of the form~(\ref{eq:nPQSS_state}),
it can be readily shown that
$\rho_n$ is a state for $n$-party PSSS,
by the same way as the proof of Theorem~\ref{Thm:SS_states}.
\end{proof}
We call the state in~(\ref{eq:nPQSS_state}) the $n$-party secret-sharing state.
Remark that, as in Theorem~\ref{Thm:SS_states} and Theorem~\ref{Thm:block_SS},
any quantum state of the form of
$2^n\times 2^n$ block matrix,
whose block entries vanish if they are in the rows or columns of odd parity,
has well-chosen $2^{n-1}-1$ off-diagonal block entries of the trace norm $1/2^{n-1}$
if and only if the state is an $n$-party secret-sharing state.

We now consider the bipartite distillability of the $n$-party secret-sharing states.

\begin{Thm}\label{Thm:n_distillability}
Any $n$-party secret-sharing state $\rho_n$ is bipartite distillable
for all bipartite splits of the $n$ parties.
\end{Thm}
\begin{proof}
We use the mathematical induction on $n\ge 3$.
Then if $n=3$ then this theorem is true by Theorem~\ref{Thm:distillability}.
We assume that this theorem is true for $(n-1)$-party secret-sharing states.
Let $P$ be an arbitrary bipartite split \mbox{$I_0$-$I_1$} of
the $n$ parties, $\{A_1A_1', A_2A_2',\ldots, A_nA_n'\}$.
Then for $A_jA_j', A_kA_k' \in I_0$
applying $\mathrm{CNOT}_{A_jA_k}$ to $\rho_n$
and performing the projective measurement on system $A_j$
with respect to the standard basis $\{\ket{0},\ket{1}\}$,
the resulting state becomes
an $(n-1)$-party secret-sharing state.
By the induction hypothesis,
the state is bipartite distillable for the split, \mbox{$I_0-\{A_jA_j'\}$-$I_1$},
and hence $\rho_n$ is also bipartite distillable for \mbox{$I_0$-$I_1$}.
This completes the proof.
\end{proof}

{\em Example.---} 
We construct a class of secret-sharing states
in a similar way to one presented in~\cite{HHHO1}.
Consider the following state:
\begin{eqnarray}
\rho &=& a_0 \ket{\psi_0}\bra{\psi_0} \otimes \sigma_0
+ a_1 \ket{\psi_1}\bra{\psi_1} \otimes \sigma_1
\nonumber\\
&&+ a_2 \ket{\psi_2}\bra{\psi_2} \otimes \sigma_2
+ a_3 \ket{\psi_3}\bra{\psi_3} \otimes \sigma_3,
\label{eq:example0}
\end{eqnarray}
where
\begin{eqnarray}
\ket{\psi_0} &=& \frac{1}{2}(\ket{000}+\ket{011}+\ket{101}+\ket{110}),\nonumber\\
\ket{\psi_1} &=& \frac{1}{2}(\ket{000}+\ket{011}-\ket{101}-\ket{110}),\nonumber\\
\ket{\psi_2} &=& \frac{1}{2}(\ket{000}-\ket{011}+\ket{101}-\ket{110}),\nonumber\\
\ket{\psi_3} &=& \frac{1}{2}(\ket{000}-\ket{011}-\ket{101}+\ket{110}),
\label{eq:psi_i}
\end{eqnarray}
and the states $\sigma_j$ have support on orthogonal subspaces.
Then one can readily verify that $\rho$ is a 3-party
secret-sharing state, since $\|a_0\sigma_0\pm a_1\sigma_1\pm
a_2\sigma_2 \pm a_3\sigma_3\|_1=1$.

As in Ref.~\cite{HHHO1},
one can find a secret-sharing state
which can have an arbitrarily small amount of
bipartite distillable entanglement for a certain split.
In order to find such a state,
take $a_0=a_1$ and $a_2=a_3$ such that $a_1+a_2=1/2$,
and
\begin{eqnarray}
\sigma_0 &=& \rho_s \otimes \ket{0}\bra{0}, \nonumber \\
\sigma_1 &=& \rho_a \otimes \ket{0}\bra{0}, \nonumber \\
\sigma_2 &=& \rho_a \otimes \ket{1}\bra{1}, \nonumber \\
\sigma_3 &=& \rho_s \otimes \ket{1}\bra{1},
\label{eq:rho_i}
\end{eqnarray}
where $\rho_s$ and $\rho_a$ and two extreme $d\otimes d$ Werner states
\begin{eqnarray}
\rho_s=\frac{2}{d^2+d}P_{sym}= \frac{\mathcal{I}+\mathcal{F}}{d^2+d}, \nonumber \\
\rho_a =\frac{2}{d^2-d}P_{as} = \frac{\mathcal{I}-\mathcal{F}}{d^2-d}
\label{eq:Werner}
\end{eqnarray}
with the identity operator $\mathcal{I}$ on the $d\otimes d$
system and the flip operator
$\mathcal{F}=\sum_{i,j=0}^{d-1}\ket{ij}\bra{ji}$. Then we have
$\|\rho^{T_{AA'}}\|_{1}= {(d+2)}/{d}$. Therefore, since the
log-negativity is an upper bound of the distillable
entanglement~\cite{VW}, the bipartite distillable entanglement for
the split \mbox{$AA'$-$BB'CC'$} can be arbitrarily small by
increasing $d$. Nevertheless, the state is always a secret-sharing
state for any $d$.

In conclusion, we have presented necessary and sufficient conditions for secret-sharing states,
and have also shown that any secret-sharing state is bipartite distillable
for its all bipartite splits,
and hence the states could be distillable into the GHZ state.
We have furthermore generalized our results into multipartite cases,
and have exhibited a class of secret-sharing
states, which have an arbitrarily small amount of
bipartite distillable entanglement for a certain split.

D.P.C. was supported by the Korea Science and Engineering
Foundation (KOSEF) grant funded by the Korea government (MOST)
(No.~R01-2006-000-10698-0), J.S.K was supported by Alberta's
informatics Circle of Research Excellence (iCORE), and S.L. was
supported by the Korea Research Foundation Grant funded by the
Korean Government (MOEHRD, Basic Research Promotion Fund)
(KRF-2007-331-C00049).


\end{document}